\documentclass[runningheads,envcountsame]{llncs}
\usepackage[utf8]{inputenc}
\usepackage[T1]{fontenc}
\usepackage{amsmath,amssymb}
\usepackage{xspace}
\usepackage[usenames,dvipsnames,table]{xcolor}
\usepackage{wrapfig,graphicx}
\DeclareGraphicsExtensions{.pdf,.png,.eps}
\graphicspath{{figs/}}
\usepackage[inline]{enumitem}

\let\doendproof\endproof\renewcommand\endproof{~\hfill$\qed$\doendproof}
\usepackage{bigdelim,multirow}

\usepackage{mathtools}
\DeclarePairedDelimiter\set{\{}{\}}
\DeclarePairedDelimiter\abs{\lvert}{\rvert}
\def\Oh{\ensuremath{\mathcal{O}}}
\def\Pl{\ensuremath{P_\mathrm{l}}\xspace}
\def\Pr{\ensuremath{P_\mathrm{r}}\xspace}
\def\Pt{\ensuremath{P_\mathrm{t}}\xspace}
\def\Pb{\ensuremath{P_\mathrm{b}}\xspace}
\def\R{\ensuremath{\mathcal{R}}\xspace}
\def\P{\ensuremath{\mathcal{P}}\xspace}
\def\Sdone{\ensuremath{\mathcal{S}_1^\checkmark\!}\xspace}
\def\vWest{\ensuremath{v_\mathrm{W}}\xspace}
\def\vSouth{\ensuremath{v_\mathrm{S}}\xspace}
\def\vEast{\ensuremath{v_\mathrm{E}}\xspace}
\def\vNorth{\ensuremath{v_\mathrm{N}}\xspace}
\def\eps{\ensuremath{\varepsilon}\xspace}
\def\RecDual{\texttt{RecDual}}
\def\RepEx{\texttt{RepEx}}

\usepackage{url}
\usepackage{cite}
\usepackage[hidelinks=true,
			bookmarks=true,
			bookmarksopen=true,
			bookmarksopenlevel=2,
			bookmarksnumbered=true,
			hyperindex=true,
			plainpages=false,
			pdfpagelabels=true,
]{hyperref} 
\usepackage[capitalise,nameinlink]{cleveref}
\crefname{enumi}{Condition}{Conditions}

\newcommand{\etal}{{et~al.}}

\newcommand{\doi}[1]{\href{https://dx.doi.org/#1}{\texttt{doi:#1}}}
\definecolor{defblue}{rgb}{0.08235294118,0.3098039216,0.537254902}
\let\emph\relax
\DeclareTextFontCommand{\emph}{\color{defblue}\em}

\title{Extending Partial Representations of Rectangular Duals with Given Contact Orientations%
\thanks{Partially supported by DFG grants Ru$\,$1903/3-1 and Wo$\,$758/11-1. 
We thank the anonymous reviewers for their helpful comments.}} 
\titlerunning{Extending Partial Representations of Rectangular Duals}

\author{Steven~Chaplick\inst{1}\orcidID{0000-0003-3501-4608} 
\and Philipp~Kindermann\inst{2}\orcidID{0000-0001-5764-7719} 
\and Jonathan~Klawitter\inst{3}\orcidID{0000-0001-8917-5269} 
\and Ignaz~Rutter\inst{4}\orcidID{0000-0002-3794-4406} 
\and Alexander~Wolff\inst{3}\orcidID{0000-0001-5872-718X}
}
\renewcommand{\orcidID}[1]{\href{https://orcid.org/#1}{\includegraphics[scale=.03]{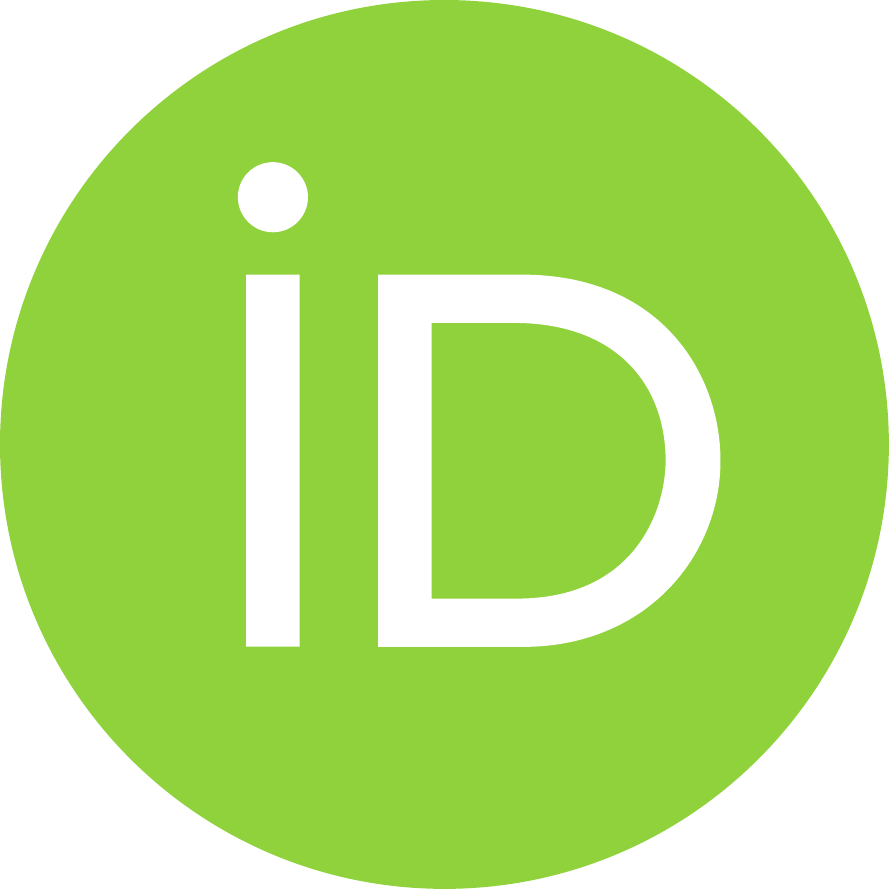}}}

\authorrunning{Chaplick et al.}
\institute{Maastricht University, Maastricht, The Netherlands  
\and
	Universität Trier, Trier, Germany
\and
	Universität Würzburg, Würzburg, Germany
\and
	Universität Passau, Passau, Germany
} 

\usepackage[disable]{todonotes}

\begin{document}

\maketitle        


\pdfbookmark[1]{Abstract}{Abstract}
\begin{abstract}
A rectangular dual of a graph $G$ is a contact representation of $G$ by axis-aligned rectangles
such that (i)~no four rectangles share a point and 
(ii)~the union of all rectangles is a rectangle.
The \emph{partial representation extension problem} for rectangular duals asks 
whether a given partial rectangular dual can be extended to a rectangular dual,
that is, whether there exists a rectangular dual where some vertices are represented by prescribed rectangles.
Combinatorially, a rectangular dual can be described by a regular edge labeling (REL), 
which determines the orientations of the rectangle contacts.

We describe two approaches to solve the partial representation extension problem for rectangular duals with given REL.
On the one hand, we characterise the RELs that admit an extension, which leads to a linear-time testing algorithm.
In the affirmative, we can construct an extension in linear time.
This partial representation extension problem can also be formulated as
a linear program (LP).  
We use this LP to solve the \emph{simultaneous
  representation problem} for the case of rectangular duals when each input graph is given together with a REL.
\end{abstract}

\section{Introduction}

A \emph{geometric intersection representation} of a graph~$G$ is
a mapping $\R$ that assigns to each vertex $w$ of $G$ a geometric object
$\R(w)$ such that two vertices~$u$ and~$v$ are adjacent in~$G$ if and
only if $\R(u)$ and $\R(v)$ intersect.  
In a \emph{contact representation} we further require that,
for any two vertices $u$ and $v$, the objects
$\R(u)$ and $\R(v)$ have disjoint interiors.
The \emph{recognition problem} asks whether a given graph admits
an intersection or contact representation whose sets have a specific
geometric shape.  Classic examples are interval
graphs~\cite{BL76}, where the objects are intervals of $\mathbb{R}$,
or coin graphs~\cite{Koe36}, 
where the objects are interior-disjoint disks in the plane.
The \emph{partial representation extension problem} is a natural
generalization of this question where, 
for each vertex~$u$ of a given subset of the vertex set, 
the geometric object is already prescribed, 
and the question is whether this partial representation can
be extended to a full representation of the input graph.  
In the last decade the partial representation extension problem has been intensely
studied for various classes of intersection graphs, such as
(unit or proper) interval graphs~\cite{kkosv17,kkorsv17},
circle graphs~\cite{CFK19}, trapezoid graphs~\cite{KW17}, as well as
for contact representations~\cite{CDKMS14} and bar-visibility
representations~\cite{CGGKL18}.

A different generalization is the \emph{simultaneous representation
  problem}, where, given several input graphs $G_1, \ldots, G_k$, one
asks whether there exist representations~$\R_1,\dots,\R_k$
of~$G_1,\dots,G_k$ such that each vertex~$v$ contained in~$G_i$ and
in~$G_j$ satisfies~$\R_i(v) = \R_j(v)$, i.e., any two representations
coincide on the shared vertices.  Most frequently, this problem is
studied in the \emph{sunflower case}, where one additionally assumes
that the pairwise intersection of any two graph~$G_i,G_j$
with~$i \ne j$ is the same subgraph~$S$, which is usually called the
\emph{shared graph}.  The question is equivalent to asking whether
there exists a representation of~$S$ that simultaneously extends to
each of the input graphs~$G_1,\dots,G_k$.  Simultaneous representation
problems have long been studied for planar graphs;
see~\cite{bkr-sepg-13,Rut20} for surveys.  For intersection
representations, the problem was originally introduced by Jampani and
Lubiw, who gave polynomial-time algorithms for interval
graphs~\cite{jl-sig-10} as well as for comparability and permutation
graphs~\cite{jl-srpcc-12}.  They also proved NP-completeness for
chordal graphs.  Bläsius and Rutter later improved the running time
for interval graphs to linear~\cite{br-spqoa-16}.  Recently, Rutter et
al.\cite{RSSV19} gave efficient algorithms for proper and unit
interval graphs.  Previous work on simultaneous contact
representations has focused on representing planar graphs and their
duals, for example, with triangles in the plane~\cite{GLP12} or with
boxes in~3D~\cite{AEKPTU15}.

\paragraph{Rectangular duals.}

In this paper we consider the partial representation extension problem for the following type of representation.
A \emph{rectangular dual} of a graph~$G$ is a contact representation~$\R$ of $G$ by axis-aligned rectangles such that (i)~no four rectangles
share a point and (ii)~the union of all rectangles is a rectangle; 
see \cref{fig:dual}. 
We observe that~$G$ may admit a rectangular dual only if it is planar and internally triangulated. 
Furthermore, a rectangular dual can always be augmented with four additional
vertices (one on each side) so that only four rectangles touch the
outer face of the representation.
It is customary that the four vertices on the outer face are denoted 
by $\vSouth$, $v_\mathrm{W}$, $\vNorth$, and $\vEast$ corresponding to the geographic directions,
and to require that $\R(\vWest)$ is the leftmost rectangle, $\R(\vEast)$ is
rightmost, $\R(\vSouth)$ is bottommost, and $\R(\vNorth)$ is topmost; see \cref{fig:dual}.
We call these vertices the \emph{outer vertices} and the remaining
ones the \emph{inner vertices}.
It is known that a plane internally-triangulated graph has a representation 
with only four rectangles touching the outer face
if and only if its outer face is a 4-cycle and it has no separating triangles,
that is, a triangle whose removal disconnects the graph~\cite{KK85}.
Such a graph is called a \emph{properly-triangulated planar (PTP)} graph. 
Kant and He~\cite{KH97} have shown that a rectangular dual of a given
PTP graph~$G$ can be computed in linear time.

\begin{figure}[tb]
  \centering
  \includegraphics[width=\linewidth]{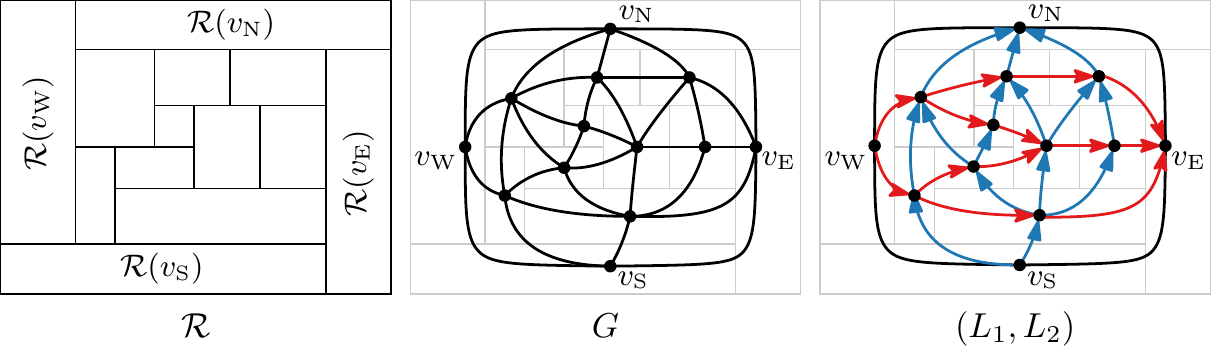}
  \caption{A rectangular dual $\R$ for the graph $G$; the REL $(L_1,L_2)$
    induced by~$\R$.}
  \label{fig:dual}
\end{figure}

Historically, rectangular duals have been studied due to their applications in architecture~\cite{Ste73}, 
VLSI floor-planning~\cite{LL84,YS95}, and cartography~\cite{GS69}.
Besides the question of an efficient construction algorithm~\cite{KH97},
other problems concerning rectangular duals are area minimization~\cite{BGPV08}, sliceability~\cite{KS15}, and area-universality,
that is, rectangular duals where the rectangles can have any given areas~\cite{EMSV12}.
The latter question highlights the close relation between rectangular duals and rectangular cartograms.
Rectangular cartograms were introduced in 1934 by Raisz~\cite{Raisz34} and
combine statistical and geographical information in thematic maps, 
where geographic regions are represented as rectangles and scaled in proportion to some statistic.
There has been lots of work on efficiently computing rectangular cartograms~\cite{HKPS04,vKS07,BSV12};
Nusrat and Kobourov~\cite{NK16} recently surveyed this topic.
As a dissection of a rectangle into smaller rectangles, 
a rectangular dual is also related to other types of dissections, 
for example with squares~\cite{BSST40} or hexagons~\cite{DGHKK12};
see also Felsner's survey~\cite{Fel13}.

\paragraph{Regular edge labelings.}
 
The combinatorial aspects of a contact representation of a graph~$G$ 
can often be described with a coloring and orientation of the edges
of~$G$.  For example, Schnyder woods describe contact representations
of planar graphs by triangles~\cite{dFdMR94}.
Such a description also exists for contact representations by rectangles, for example for
triangle-free rectangle arrangements~\cite{KNU15} or rectangular duals~\cite{KH97}.
More precisely, a rectangular dual $\R$ gives rise to a 2-coloring 
and an orientation of the inner edges of $G$ as follows.  
We color an edge $\set{u, v}$ blue if the contact between $\R(u)$ 
and $\R(v)$ is a horizontal line segment, and we color it red otherwise.  
We orient a blue edge $\set{u, v}$ as $(u,v)$ if $\R(u)$ lies below $\R(v)$, 
and we orient a red edge $\set{u, v}$ as $(u,v)$ if $\R(u)$ 
lies to the left of $\R(v)$; see \cref{fig:dual}.  
The resulting coloring and orientation has the following properties:
\begin{figure}[tb]
	\centering
	\includegraphics{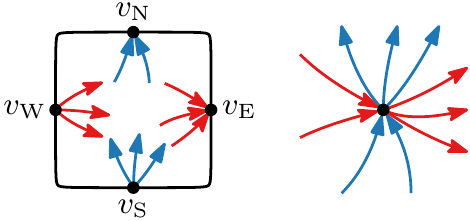} 
	\caption{Edge order at the four outer vertices and an inner vertex.}
\end{figure}
\begin{enumerate}[label=(\arabic*)]
\item All inner edges incident to $\vWest$, $\vSouth$, $\vEast$, and $\vNorth$ are red outgoing, blue
  outgoing, red incoming, and blue incoming, respectively. 
\item The edges incident to each inner vertex form four counterclockwise
  ordered non-empty blocks of red incoming, blue incoming, red outgoing, and blue
  outgoing, respectively.
\end{enumerate}
A coloring and orientation with these properties is called a
\emph{regular edge labeling (REL)} or \emph{transversal structure}.
We let $(L_1,L_2)$ denote a REL,
where $L_1$ is the set of blue edges and $L_2$ is the set of red edges.
Let~$L_1(G)$ and~$L_2(G)$ denote the two subgraphs of~$G$ induced by~$L_1$
and~$L_2$, respectively.
Note that both~$L_1(G)$ and $L_2(G)$ are st-graphs, that is,
directed acyclic graphs with exactly one source and exactly one sink. 
It is well known that every PTP graph admits a REL and thus a rectangular dual~\cite{KH97}. 
A rectangular dual $\R$ \emph{realizes} a REL $(L_1, L_2)$ if the
REL induced by~$\R$ is $(L_1, L_2)$.
Note that while a rectangular dual
uniquely defines a REL, there exist different rectangular duals
that realize any given REL. 

Kant and He~\cite{KH97} introduced RELs
and described two linear-time algorithms that compute a REL for a given PTP graph;
one algorithm is based on edge contractions, the other is based on canonical orderings.
They then use the REL to construct in linear time a rectangular dual that realizes this REL
and where the coordinates are all integers.

\paragraph{Partial rectangular duals.}

For a graph $G$, let $E(G)$ denote the set of edges
and~$V(G)$ the set of vertices of $G$.  
Let~$U \subseteq V(G)$.
Then~$G[U]$ denotes the subgraph of~$G$ induced by~$U$.
The pair~$(U, \P)$ is a \emph{partial rectangular dual} of~$G$
if~$\P$ is a contact representation of~$G[U]$ 
that maps each $u \in U$ to an axis-aligned rectangle $\P(u)$.
We call the vertices in~$U$ \emph{fixed} and for~$u \in U$ we call~$\P(u)$ a \emph{fixed rectangle}.
We further define $\abs{ \P } = \abs{ U }$.
For the sake of readability, we refer from now on to a partial rectangular dual $(U, \P)$ simply with $\P$ and consider the domain $U$ of~$\P$ as implicitly given.

For a given graph $G$ and a partial rectangular dual $\P$,
the \emph{partial rectangular dual extension problem}
asks whether $\P$ can be extended to a rectangular dual $\R$ of $G$.
In particular, for such an extension $\R$ and each fixed vertex $u$, we require that $\P(u) = \R(u)$.
In this paper, we study the variant of this problem where we are not only given $G$ and $\P$,
but also a REL $(L_1, L_2)$ of $G$ and 
ask whether there is an extension $\R$ of $\P$ that realizes $(L_1, L_2)$. 

Closely related work includes partial representation extension
of segment contact graphs~\cite{CDKMS14} and bar-visibility representations~\cite{CGGKL18}.  
Both problems are NP-complete.  
However, the hardness reductions
crucially rely on low connectivity for choices in the planar embedding.  
Since PTP graphs are triconnected, they have a unique planar embedding and hence these results cannot be easily transferred.

\paragraph{Contribution and outline.}

Our first contribution is a linear program (LP) in the form of a system of difference constraints
to compute rectangular duals for PTP graphs with given RELs.
We show how to use the LP (i)~to construct a rectangular dual,
(ii)~to solve the partial representation extension problem, and 
(iii)~to solve the respective simultaneous representation problem
in quadratic time; see \Cref{sec:lp}.

We then give a characterization of RELs that admit an extension of a given partial rectangular dual
via the existence of what we will call a \emph{boundary path set}; see \Cref{sec:characterization}.
Next, we provide an algorithm that constructs a boundary path set (if possible) 
as well as an algorithm that computes a representation extension from a boundary path set. 
Both algorithms run in $\Oh(nh)$ time, where $n = \abs{ V(G) }$ and $h = \abs{ \P }$, and are detailed in \Cref{sec:algorithms}.
Finally, we show that by checking only for the existence of a boundary
path set, but not explicitly constructing one,
we can solve the partial representation extension problem
in linear time; see \Cref{sec:linear-time}.
Our algorithms use the above-mentioned algorithm of Kant and He~\cite{KH97} as a subroutine.
We summarize our main contribution as follows.

\newcommand{\thmtext}{The partial representation extension problem for
  rectangular duals with a fixed regular edge labeling can be solved
  in linear time.  For yes-instances, an explicit rectangular dual can
  be constructed within the same time bound.}
\begin{theorem}\label{clm:main}
  \thmtext
\end{theorem}

\section{Linear Programming}
\label{sec:lp}

In this section, we describe how a rectangular dual of a given PTP
graph and a given REL can be computed with the help of an LP.
Felsner~\cite{Fel13} used an LP to compute square duals,
that is, a rectangular dual where each rectangle is actually a square.
While Felsner's LP can be adapted to compute rectangular duals,
we formulate our LP differently such that 
we can also use it for the partial representation extension problem
and the simultaneous representation problem.  For the same reason,
our LP does not compete with the (linear-time) algorithm of 
Kant and He~\cite{KH97}, which solves the task faster than our LP.

Our LP is a so-called \emph{system of difference constraints (SDC)}.
This means that, if we write the LP in the standard form $Ax \le b$, 
every entry of the matrix $A$ is in $\{-1,0,1\}$ and in each row
of~$A$ at most one entry is a~$1$ and at most one entry is a~$-1$.
The advantage of an SDC is that the Bellman--Ford algorithm can be
used to find a solution (if one exists) in $\Oh(n^2 + nm)$ time,
where $n$ is the number of variables and $m$ is the number of
constraints~\cite{CLRS09}.

Let $G$ be a PTP graph, 
let $(L_1, L_2)$ be a REL of~$G$, and let $\eps > 0$.
We call our LP \RecDual$(G, (L_1, L_2), \eps)$.
We first describe the variables and then the constraints.
We associate four variables with each vertex~$u$ of~$G$.
The variables~$x_{1,u}$ and~$x_{2,u}$ denote the x-coordinates of the
left side and the right side of~$\R(u)$, respectively,
and the variables~~$y_{1,u}$ and~$y_{2,u}$ denote the y-coordinates of
the bottom side and the top side of~$\R(u)$, respectively.
In what follows, we treat only the constraints regarding the
x-variables.  The constraints regarding the y-variables are
analogous.  There are no constraints regarding both types of variables.
We require that each rectangle has width and height at least~\eps,
that~is,
\begin{align*}
  x_{2,u} - x_{1,u}  &\ge \eps && \text{for each vertex $u$ of $G$.}\\
  \intertext{We have two types of constraints for the edges. 
  \newline\indent
  First, for every edge $(u, v) \in E(L_2(G))$, we ensure that the
  left side of~$\R(v)$ touches the right side of~$\R(u)$; see
  \cref{fig:LPexample}(a).  In other words,}
  x_{2,u} - x_{1,v} &= 0 && \text{for each edge } (u, v) \in E(L_2(G)). \\
  \intertext{(We treat the two edges $(\vSouth, \vWest)$ and $(\vEast,
  \vNorth)$ as edges of $L_1(G)$ and the two edges $(\vWest, \vNorth)$
  and $(\vSouth, \vEast)$ as edges of $L_2(G)$; see \cref{fig:dual}.
  Thus, the lower left
  corner of the rectangular dual belongs to $\R(\vSouth)$.)
  \newline\indent
  Second, for every edge $(u, v) \in E(L_1(G))$, we enforce that the
  rectangles $\R(u)$ and $\R(v)$ overlap horizontally; see
  \cref{fig:LPexample}(b).  To this end, for a vertex~$u$ in~$G$, let~$v$
  and~$v'$ be the (clockwise) first and last outgoing neighbors of~$u$
  in~$L_1(G)$.  (They do not necessarily have to be distinct.)  Then,}
  x_{2,v} - x_{1,u} &\ge \eps && 
                                 \text{and}\\
  x_{2,u} - x_{1,v'} &\ge \eps && \text{for each vertex $u$ of $G$ with
                                  outgoing neighbors $v$ and $v'$.}
\end{align*}  

\begin{figure}[tb]
  \centering
  \includegraphics{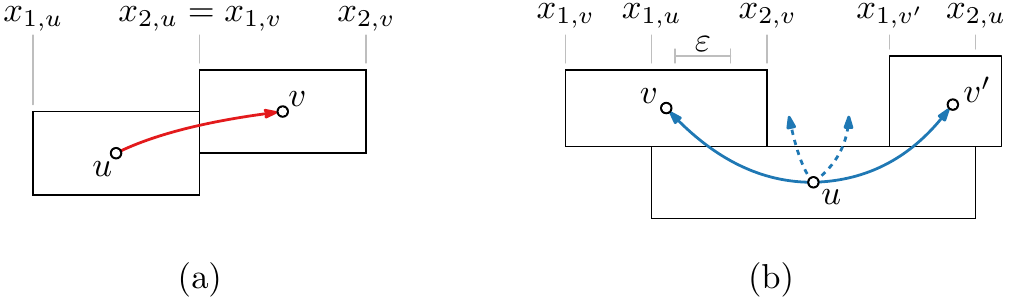}
  \caption{The relation between variables of the SDC such that the
    edges of the REL are represented correctly. }
  \label{fig:LPexample}
\end{figure}

As a result, $\R(v)$ and $\R(v')$ overlap with $\R(u)$ horizontally
by at least~\eps.
Rectangles corresponding to other outgoing neighbors of $u$ 
overlap with $\R(u)$ because they lie between $\R(v)$ and $\R(v')$ by the first type of constraint.
We have analogous inequalities for the first and last {\em incoming} neighbor of~$u$ in~$L_1(G)$.

If the given graph~$G$ has $n$ vertices, our SDC has $\Oh(n)$ variables
and constraints, hence we can solve it in $\Oh(n^2)$ time using
Bellman--Ford.  This, however, tells us only whether a solution
exists.  In our case, this is not interesting since if we have a REL
$(L_1, L_2)$ of $G$, then we know that a rectangular dual exists.
We now show how to minimize the area and the perimeter of a
rectangular dual, by reducing the optimization problem to the decision
problem.  (Kant and He~\cite{KH97} do not show this explicitly, but
for a given REL, their linear-time algorithm yields a rectangular
dual of minimum area and perimeter.  They posed the open question
whether it is possible to find a rectangular dual of minimum area or
perimeter for a given PTP graph if the REL is not fixed.  In any case,
our slower LP is more flexible in that we can easily add additional
constraints concerning, e.g., the size or position of rectangles.)

Note that the horizontal and vertical dimensions of a rectangular dual
can be treated independently if the REL is fixed.  Hence, it suffices
to independently minimize the width and the height of the rectangular
dual.  If we set $x_{1,\vSouth} = 0$, then the
x-coordinate~$x_{2,\vNorth}$ of the right side of~\vNorth equals the
width of the rectangular dual.  Hence, we need to
minimize~$x_{2,\vNorth}$.  To this end, we do a binary search
on~$x_{2,\vNorth}$ and solve the SDC for each value of~$x_{2,\vNorth}$
that the binary search considers.  If we set $\eps = 1$, all
constraints use integer values.  Then the Bellman--Ford algorithm will
yield a solution (that is, a rectangular dual) with integer
coordinates.  Note that the minimum width of such a rectangular dual
is bounded by~$n$ 
(since every integer between~0 and $n$ must be used by the
x-coordinate of at least one left rectangle side and one right rectangle side). 
Hence, the binary search does $\Oh(\log n)$ steps, and we can solve the
optimization problem in $\Oh(n^2\log n)$ time. 

\paragraph{Partial Representation Extension.}
Let $\P$ be a partial rectangular dual of $G$.
With a slight modification, we can use \RecDual$(G, (L_1, L_2), \eps)$
to decide whether $\P$ admits a rectangular dual extension for $G$ and $(L_1, L_2)$.
More precisely, for each fixed vertex $u$ of $\P$, 
we set $x_{1,u}$, $x_{2,u}$, $y_{1,u}$, and $y_{2,u}$ according to the fixed rectangle $\R(u)$ of $\P$.
We call this LP \RepEx$(G, (L_1, L_2), \P, \eps)$.

We need to set $\eps$ such that rectangles can be placed between fixed rectangles
with sizes and overlaps of at least $\eps$.
Let $X$ be the set of x-coordinates $x_{1,u}$ and $x_{2,u}$ of all fixed vertices $u$.
Define $Y$ analogously.
Let $d$ be the minimum distance between any pair in $X$ or in $Y$.
Observe that in a rectangular dual of~$G$ at most $n$ distinct x-coordinates
are used for the x-coordinates of the left and right sides of the rectangles.
The analogous statement holds for the y-coordinates.
Hence, we get the following lemma.

\begin{lemma}
  \label{clm:lp:epsolution}
  If there exists any $\eps > 0$ such that
  \RepEx$(G, (L_1, L_2), \P, \eps)$ has a solution, then
  \RepEx$(G, (L_1, L_2), \P, d/n)$ has a solution as well.
\end{lemma}

The SDC \RepEx$(G, (L_1, L_2), \P, d/n)$ for the partial rectangular dual extension problem has
a worse asymptotic running time than the combinatorial methods of the subsequent sections.	
However, given $G$, $(L_1, L_2)$, and $\P$, 
the SDC is easy to generate and fast commercials solvers can be used.
Furthermore, we can solve a slightly more general problem with the LP.
For example, instead of specifying the position \textit{and} the size of each fixed rectangle,
we have the freedom to specify (or to bound) only some of these parameters.

\paragraph{Simultaneous Representation.}

Next, we describe how to use the SDC to solve the simultaneous
rectangular dual representation problem when several graphs (sharing
vertices) and their RELs are given.
The idea is to generate one SDC per graph and then, 
for each vertex shared between two graphs,
to identify the respective variables.
Therefore, if two PTP graphs $G$ and $G'$ share a vertex~$u$,
then the rectangles $\R(u)$ and $\R'(u)$ in the rectangular duals of $G$ and $G'$, respectively,
will be identical.

More formally, let $G_1,\dots,G_k$ be PTP graphs 
with RELs $(L_1, L_2)_1$, $\dots$, $(L_1, L_2)_k$, respectively.
For distinct $i$ and $j$ in $\set{1, \ldots, k}$,
let $H_{i,j}$ be the common subgraph of $G_i$ and $G_j$.
For each $i \in \set{1, \ldots, k}$, we generate the SDC \RepEx$(G_i, (L_1, L_2)_i, 1)$
and let its variables have superscript $i$.
We then merge the $k$ SDCs into a single SDC. 
To ensure a simultaneous representation, for each~$u \in V(H_{i, j})$,
we set $x^i_{1,u} = x^j_{1,u}$, $x^i_{2,u} = x^j_{2,u}$, $y^i_{1,u} = y^j_{1,u}$, and $y^i_{2,u} = y^j_{2,u}$.
The result is an SDC of size linear in the total number of vertices of
$G_1,\dots,G_k$.  This yields the following theorem.

\begin{theorem} \label{clm:simultaneous}
The simultaneous representation problem for rectangular duals
with a fixed regular edge labeling can be solved in quadratic time. 
For yes-instances,
simultaneous rectangular duals can be constructed within the same time bound.
\end{theorem}

\section{Characterization}
\label{sec:characterization}

In this section, we characterize when a given PTP graph~$G$ with REL
$(L_1, L_2)$, and partial rectangular dual~$\P$ of $G$ admits an extension~$\R$ that realizes~$(L_1, L_2)$.
Before we can explain our main idea, we require an observation and a few definitions.

We may assume
that $\vWest$, $\vSouth$, $\vEast$, and $\vNorth$ are fixed vertices of $\P$.
(Otherwise, we simply place the outer rectangles $\P(\vWest)$, $\P(\vSouth)$, $\P(\vEast)$, and
$\P(\vNorth)$ appropriately around $\P$ 
such that they touch potential neighbours in $\P$.) 
The rectangles $\P(\vWest)$, $\P(\vSouth)$, $\P(\vEast)$, and
$\P(\vNorth)$ thus form a \emph{frame} with the area inside
partially covered and partially uncovered. 
To make the question of whether this uncovered area can be filled with the rectangles of non-fixed vertices
more accessible, we subdivide the uncovered area into smaller parts
and then try to fill them one by one.
More precisely and as illustrated in \cref{fig:strips},
we draw a vertical segment through the vertical sides of each fixed rectangle of an inner vertex
until another fixed rectangle is hit.
This divides the uncovered area inside the frame into (non-empty) \emph{vertical strips}.
We call the fixed rectangles bounding a vertical strip $S$ from below and above
the \emph{start} and \emph{end} rectangles of~$S$, respectively. 
We define \emph{horizontal strips} symmetrically. 
The start and end rectangle of a horizontal strip are to its left and right, respectively.

\begin{figure}[b]
  \centering
  \includegraphics{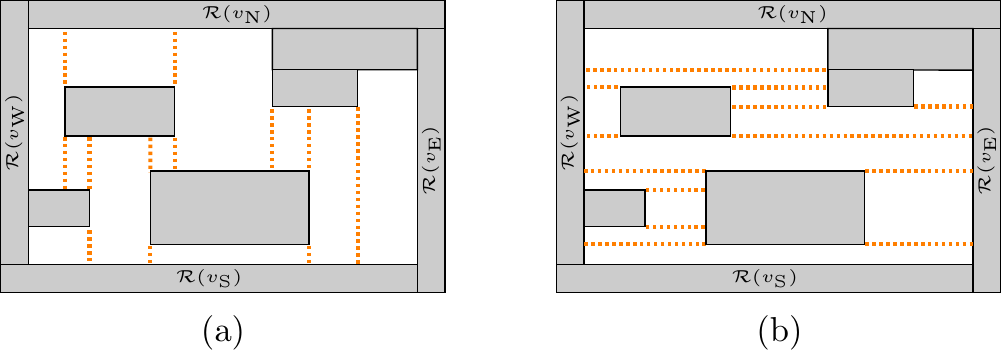}
  \caption{Dissection of the interior of the frame into (a) vertical and (b) horizontal strips.}
  \label{fig:strips}
\end{figure}

The idea for our characterization is as follows.
Consider an extension~$\R$ of~$\P$, where the strips are now filled with rectangles.
The vertical strips thus naturally induce subgraphs in $L_1(G)$ containing the vertices
that intersect it plus their start and end rectangle. 
Together, these subgraphs cover the whole of $L_1(G)$.
In particular, for a vertical strip $S$ with start rectangle~$\P(u)$ and end rectangle~$\P(v)$, 
the outer face of its induced subgraph 
consists of a path containing the rectangles along the left side of $S$
and a path along the right side of $S$.
The idea is that, even with $\R$ not known, we have to be able to
cover~$L_1(G)$ (and~$L_2(G)$) with subgraphs defined by pairs of boundary paths.
We now make this precise.

For two paths~$P$ and~$P'$ in~$L_1(G)$, we
write $P \preceq P'$ if no vertex of $P$ lies to the right of~$P'$, i.e.,
there is no path from a vertex in $P'$ to a vertex in $P\setminus P'$ in $L_2(G)$.
Let $S$ be a vertical strip with start rectangle~$\P(u)$ and end rectangle~$\P(v)$.
A \emph{boundary path pair} of $S$ is a pair of paths
$\langle\Pl(S), \Pr(S)\rangle$ from $u$ to $v$ in $L_1(G)$ 
such that $\Pl(S) \preceq \Pr(S)$
and the only fixed vertices in $V(\Pl(S) \cup \Pr(S))$ are $u$ and $v$; see \cref{fig:pathpairs}(a).
Based on the boundary path pair of $S$, 
we define $S(L_1(G))$ as the maximal subgraph of $L_1(G)$ that has precisely $\Pl(S)$ and $\Pr(S)$ as the boundary of the outer face.
The definitions for horizontal strips, where we order paths $\Pb(S)$ and $\Pt(S)$ from bottom to top, are analogous.

\begin{figure}[tb]
  \centering
  \includegraphics{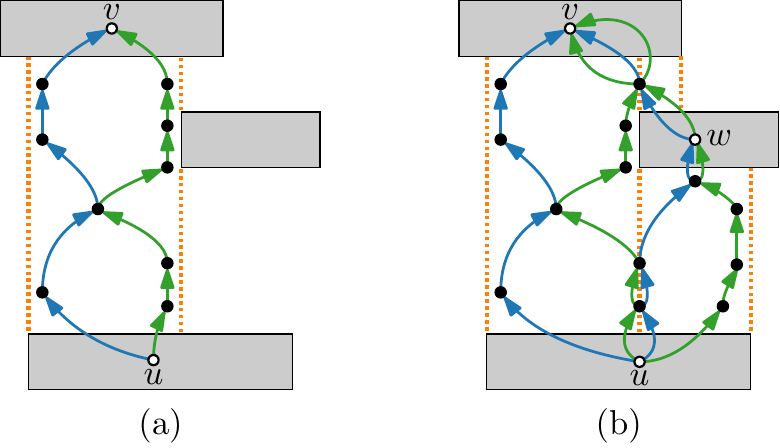}
  \caption{(a)~A boundary path pair for the strip with start rectangle $\R(u)$ and
    end rectangle $\R(v)$. (b)~Neighboring strips can have overlapping boundary
    paths.}
  \label{fig:pathpairs}
\end{figure}

Let $\mathcal{S}_1$ and $\mathcal{S}_2$ be the sets of vertical and horizontal strips, respectively.
We define a \emph{boundary path set} of a REL $(L_1, L_2)$ as a set of boundary path pairs,
one for each strip in $\mathcal{S}_1$ and $\mathcal{S}_2$, that satisfy
the following properties (see \cref{fig:pathpairs}(b)):
\begin{enumerate}[label=(B\arabic*),leftmargin=*]
  \item For strips $S$ and $S'$ in $\mathcal{S}_1$ with $S$
    left of~$S'$, it holds that $\Pr(S) \preceq \Pl(S')$.
    \label{enum:left}
  \item For strips $S$ and $S'$ in $\mathcal{S}_2$ with $S$
    below~$S'$, it holds that $\Pt(S) \preceq \Pb(S')$.
    \label{enum:below}
  \item The vertical strips cover $L_1(G)$, and the horizontal strips
    cover $L_2(G)$, that is, $\bigcup_{S \in \mathcal{S}_1} S(L_1(G)) = L_1(G)$
    and $\bigcup_{S \in \mathcal{S}_2} S(L_2(G)) = L_2(G)$.
    \label{enum:cover}
\end{enumerate}
Note that boundary paths of neighboring strips may overlap.

An extension $\R$ of $\P$ directly induces a boundary path set;
for each, say, vertical strip we simply walk through the rectangles along its left and right boundary
to find its boundary path pair.   
In the following, we show that the converse is also true.

\begin{theorem}\label{clm:characterization}
Let $G$ be a PTP graph and let $\P$ be a partial rectangular dual of $G$.
A REL $(L_1, L_2)$ of $G$ admits an extension of $\P$
if and only if $(L_1, L_2)$ admits a boundary path set. 
\end{theorem}
\begin{proof}
Suppose that $(L_1, L_2)$ admits a boundary path set.  We show how to
use this set to construct an extension of~\P. 

Let $S$ be a vertical strip, let~$S'$ be a horizontal strip, and assume
that $B = S \cap S'$ is nonempty. We call $B$ a \emph{box}.
All such boxes together with all fixed rectangles form a rectangle.
We now fill the boxes from the bottom-left to the top-right.

The paths in the pairs $\langle\Pl(S), \Pr(S)\rangle$ and
$\langle\Pb(S'), \Pt(S')\rangle$ pairwise intersect in single vertices
$v_{\mathrm{l}\mathrm{b}}$, $v_{\mathrm{l}\mathrm{t}}$,
$v_{\mathrm{r}\mathrm{b}}$, $v_{\mathrm{r}\mathrm{t}}$.
Note that some or even all of these four vertices may coincide. 
Let $G_B$ be the subgraph of~$G$ whose outer cycle is formed by the
boundary path pairs between these vertices; see \cref{fig:boxes}(a).
If we enclose~$G_B$ appropriately with a 4-cycle, we get a PTP graph with a REL
and can apply the
algorithm of Kant and He~\cite{KH97}
to compute a rectangular dual~$\R_B$ of~$G_B$.

\begin{figure}[t]
  \centering
  \includegraphics{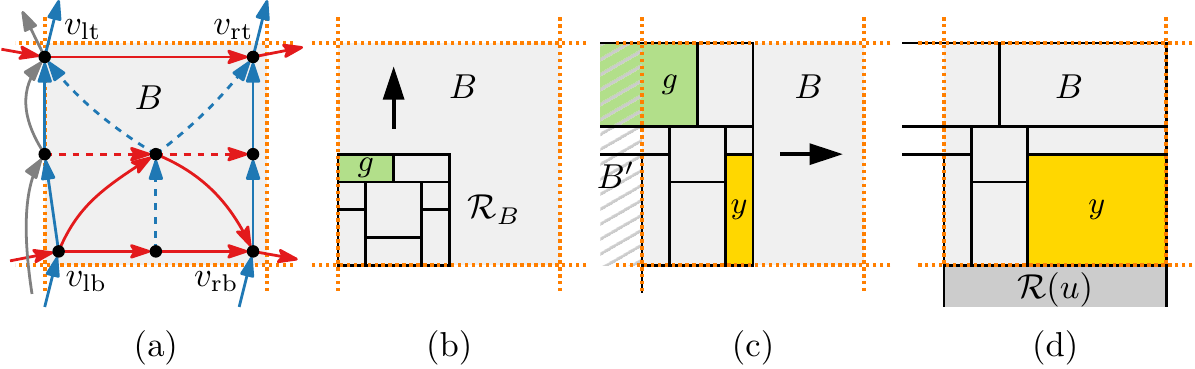}
  \caption{(a) Graph $G_B$ for a box $B$; (b)~representation $\R_B$
    for $G_B$; (c)~adjusting the left boundary of~$\R_B$ to~$\R_{B'}$;
    and (d)~adjusting the bottom boundary to~$\R(u)$.}
  \label{fig:boxes}
\end{figure}

By the order in which we fill boxes, we have already treated those
immediately to the left and below~$B$; either of them may also be
a fixed rectangle.  Without loss of generality, we assume
that there is a box~$B'$ that touches~$B$ from the left and a fixed
rectangle that touches~$B$ from below.  

First, we modify~$\R_B$ such that it fits to the rectangular
dual~$\R_{B'}$ that is drawn inside of~$B'$.  
Property~\ref{enum:left} of a boundary path set ensures that the rectangles
in~$\R_{B'}$ that are adjacent to the right side of~$\R_{B'}$ are
``compatible'' to the rectangles in~$\R_B$ that are adjacent to the
left side of~$\R_B$.  Hence, starting with a tiny version of $\R_B$
placed in the lower left corner of~$B$, we can stretch~$\R_B$
vertically along suitable horizontal cuts such that, for every vertex~$u$ in~$V(G_{B'}) \cap V(G_B)$,
the left piece of~$\R(u)$ (in~$B'$) and the right piece of~$\R(u)$
(in~$B$) fit together; see the green rectangle $g$ in
\cref{fig:boxes}(b)--(c).

Now suppose that, for some fixed vertex $u$, the fixed
rectangle~$\P(u)$ bounds~$B$ from below.  Property~\ref{enum:cover} of a
boundary path set ensures that if we stretch~$\R_B$ horizontally along some vertical cut, then
we have the correct horizontal contacts with~$\P(u)$; see the
yellow rectangle $y$ in \cref{fig:boxes}(c)--(d).

Finally, note that property~\ref{enum:cover} ensures that, at the
end of this construction, every vertex of~$G$ is represented by a
rectangle in~\R.
\end{proof}

Note that the existence of a boundary path set for a REL $(L_1, L_2)$ 
does not depend on the numeric values of the x- and y-coordinates of the fixed rectangles in $\P$.
In fact, we get the following corollary.
\begin{corollary} \label{clm:coordinateorder}
Let $G$ be a PTP graph and let $\P$ be a partial rectangular dual of~$G$.
Whether a REL $(L_1, L_2)$ of $G$ admits an extension of $\P$
depends only on the order of the first and second x-coordinates 
and the order of the first and second y-coordinates of the fixed rectangles in $\P$.
\end{corollary}

\begin{figure}[tb]
	\centering 
	\includegraphics{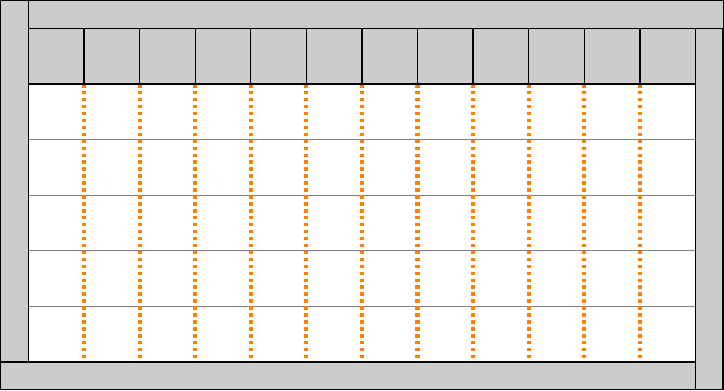}
	\caption{A rectangular dual with a boundary path set of size $\Omega(nh)$ but only five unfixed vertices.} 
	\label{fig:largeBoundaryPathSet}
\end{figure}
We close this section with an observation about the potential size of
boundary path sets.
As we have noted above, a vertex may lie on multiple boundary paths;
in fact, it may even lie on all of them as the example in \cref{fig:largeBoundaryPathSet} shows.
Hence, the size of a boundary path set can be in $\Omega(nh)$, where
$n = \abs{ V(G) }$ and $h = \abs{ \P }$.

\section{Finding a Boundary Path Set}
\label{sec:algorithms}

We now show how to compute a boundary path set for a given REL $(L_1, L_2)$ and a partial representation $\P$.
The idea is as follows.
As we did for the boxes in the proof of \cref{clm:characterization}, 
we handle the vertical strips in~$\mathcal{S}_1$ from bottom-left to top-right.
When computing the boundary path pair for a vertical strip $S \in \mathcal{S}_1$, 
we want the resulting graph $S(L_1(G))$ to include all necessary vertices but otherwise as few vertices as possible.
In particular, there may be rectangles that by $(L_1, L_2)$ need to have 
their left boundary align with the left boundary of $S$ and thus need to be in $\Pl(S)$.
To make this more precise, let $\P(v_1), \P(v_2), \ldots, \P(v_k)$ be
the fixed rectangles whose right sides touch the left side of~$S$.
Let $u$ and $v$ be the fixed vertices corresponding to the start and end rectangle of $S$, respectively.
Let $x$ be a vertex that lies on a path from $u$ to $v$ in $L_1(G)$. 
Then we say $x$ is \emph{left-bounded} in~$S$
if and only if one of the following conditions applies (see
\cref{fig:leftrightbounded}(a)):
\begin{enumerate}[label=(L\arabic*),leftmargin=*]
  \item \label{cond:lb:uv}
  $x = u$ or $x = v$ and the left side of $\P(x)$ aligns with the left side of $S$;
  \item \label{cond:lb:vi}
  $(v_i, x)$, for some $i \in \set{1, \ldots, k}$, is an edge in $L_2(G)$; 
  \item \label{cond:lb:above} 
  $(y, x)$ is the leftmost outgoing edge of $y$ and the leftmost incoming edge of $x$ in $L_1(G)$, and $y$ is left-bounded;
  \item \label{cond:lb:below}
  $(x, y)$ is the leftmost outgoing edge of $x$ and the leftmost incoming edge of $y$ in $L_1(G)$, and $y$ is left-bounded.
\end{enumerate}
\Cref{cond:lb:vi} applies if $\R(x)$ has to be directly to the right of a fixed rectangle left of $S$.
\Cref{cond:lb:above,cond:lb:below} apply if the left side of $\R(x)$ has to align 
with the left side of a left-bounded rectangle $R(y)$ directly below or above, respectively. 
Note that in this case there exists also a vertex $y'$ that is right-bounded in a strip $S'$ left of $S$ and $(y', x) \in E(L_2(G))$.

\begin{figure}[htb]
  \centering
  \includegraphics{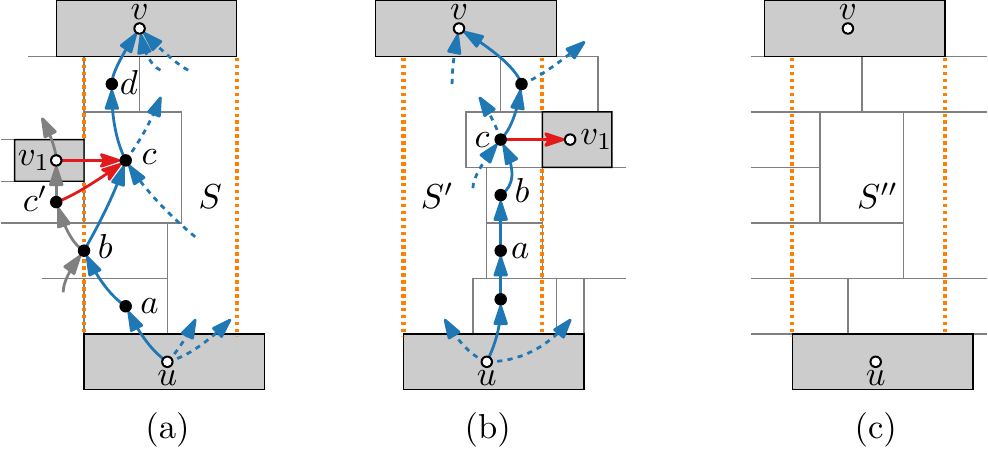}
  \caption{(a) Vertex $u$ is left-bounded in~$S$ by \cref{cond:lb:uv},
    $c$ on \cref{cond:lb:vi}, $a$ on \cref{cond:lb:above}, and $b$ is
    not left-bounded; (b)~vertex $c$ is right-bounded in~$S'$ by
    \cref{cond:rb:vi}, and $a$ and $b$ by \cref{cond:lb:below}; 
    (c)~$S''$ has neither left- nor right-bounded vertices except for
    $u$ and $v$.}
  \label{fig:leftrightbounded}
\end{figure}

Next, let $\P(v_1'), \P(v_2'), \ldots, \P(v_{k'}')$
be the fixed rectangles whose left sides touch the right side of~$S$.
Then $x$ is \emph{right-bounded} in~$S$ if and only if one of the
following conditions applies (see \cref{fig:leftrightbounded}(b)):
\begin{enumerate}[label=(R\arabic*),leftmargin=*]
  \item \label{cond:rb:uv}
  $x = u$ or $x = v$ and the right side of $\P(x)$ aligns with the right side of $S$; 
  \item \label{cond:rb:vi}
  $(x, v_i')$, for some $i \in \set{1, \ldots, k'}$, is an edge in $G_2$;
  \item \label{cond:rb:above}
  $(y, x)$ is the rightmost outgoing edge of $y$ and the rightmost incoming edge of $x$ in $L_1(G)$, and $y$ is right-bounded; 
  \item \label{cond:rb:below}
  $(x, y)$ is the rightmost outgoing edge of $x$ and the rightmost incoming edge of $y$ in $L_1(G)$, and $y$ is right-bounded.
\end{enumerate}
Note that $x$ can be both left- and right-bounded. 
Furthermore, starting from $u,v'_1,\dots,v'_{k'},v$,
these conditions can easily be checked for each strip.
Overall, we can thus find all left- and right-bounded vertices of all strips in $\Oh(n)$ time. 
\begin{theorem}\label{clm:algorithm}
Let $G$ be a PTP graph with $n$ vertices and REL $(L_1, L_2)$,
let $\P$ be a partial rectangular dual of $G$, and let $h = \abs{ \P }$.\\
In $\Oh(nh)$ time, we can decide whether $(L_1, L_2)$ admits a boundary path set 
with respect to $\P$ and, in the affirmative, compute it.
\end{theorem}
\begin{proof}
We show how to compute the boundary path pairs for vertical strips;
horizontal strips can be treated analogously.
Let $\Sdone$ be the strips in~$\mathcal{S}_1$ that
have already been processed, that is, the strips for which the left and right boundary paths have already been computed.
Let $S$ be a strip with start rectangle~$\P(u)$ and end rectangle~$\P(v)$
such that every strip left of $S$ is in~$\Sdone$.

An edge $(x, y)$ of $L_1(G)$ is \emph{suitable} if one of the following
conditions applies:
\begin{enumerate}[label=(E\arabic*),leftmargin=*]
	\item \label{cond:v}
    $y = v$;
  \item \label{cond:left}
    $y$ is a non-fixed vertex in $\Pr(S')$, where $S' \in \Sdone$
    is directly left of $S$ and $y$ is not right-bounded in~$S'$;
  \item \label{cond:new}
    $y$ is a non-fixed vertex and $(x, y)$ is not an edge of $\Sdone(L_1(G))$.
\end{enumerate}
Condition~\ref{cond:left} means that $\R(y)$ can span from $S'$
into~$S$ since it is not right-bounded in~$S'$.
Thus, in \cref{fig:leftrightbounded}(a) $(a, b)$ is suitable but $(b, c')$ is not.
Furthermore, $(d, v)$ is suitable by \cref{cond:v}, 
and $(u, a)$, $(a, b)$, and $(c, d)$ are suitable by \cref{cond:new}.  
Note that $\Pl(S)$ may only use suitable edges. 
Hence, to compute $\Pl(S)$, we can start at $u$ and 
always add the leftmost suitable outgoing edge until we reach~$v$.
It follows that if, at some point, there is no suitable edge available, 
then $(L_1, L_2)$ does not admit a boundary path set.
Taking the leftmost suitable outgoing edge ensures that $\Pl(S)$ 
passes through all left-bounded vertices in~$S$. 

We now show how to construct $\Pr(S)$, enforcing 
that all right-bounded vertices lie on $\Pr(S)$.
We thus start with the set of disjoint subpaths $P_1, P_2, \ldots, P_{k'}$ induced by $u$, the right-bounded vertices, and $v$ 
ordered from bottom to top; see \cref{fig:algo}(a).
Note that for a right-bounded vertex $x$
its rightmost outgoing edge also has to be in $\Pr(S)$, unless $x = v$,
and its rightmost incoming edge also has to be in $\Pr(S)$, unless $x = u$.
Therefore, we extend each subpath with these rightmost outgoing and
incoming edges; see \cref{fig:algo}(b).
For $i \in \set{1, \dots, k'-1}$,
we then simultaneously extend $P_i$ and $P_{i+1}$ 
by always taking the leftmost suitable outgoing and incoming edge, respectively,
but without crossing $\Pl(S)$.
If the extensions of~$P_i$ and~$P_{i+1}$ meet, we join them; see
\cref{fig:algo}(c).  Otherwise, both extensions will stop
(due to a lack of suitable edges). In this case
there is no path $\Pr(S)$, and then the REL $(L_1, L_2)$ does not
admit a boundary path set.

\begin{figure}[tb]
  \centering
  \includegraphics{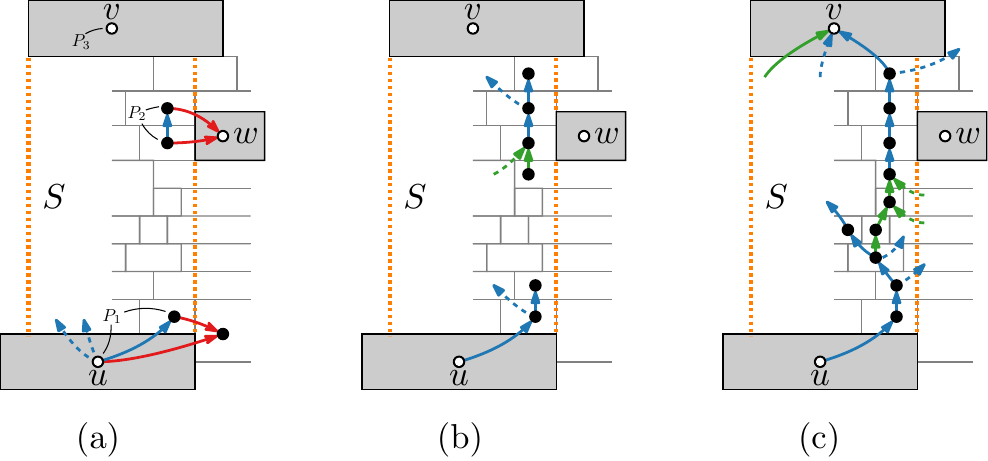}
  \caption{Computation of $\Pr(S)$ (a) starting with subpaths induced by $u$, right-bounded vertices and $v$, 
  (b) extending these subpaths with their rightmost predecessor and successor, 
  and (c) leftwards until they meet. (Extensions downwards are shown green.)}
  \label{fig:algo}
\end{figure}

Once $\Pl(S)$ and $\Pr(S)$ have been computed successfully, 
we update the edge set of $\Sdone(L_1(G))$ before processing the next strip.

The runtime is linear in the size of the boundary path set, that is, $\Oh(nh)$.
\end{proof}

\noindent
Next, we show how to obtain an extension of $\P$ from a boundary path set.

\begin{theorem}\label{clm:drawing}
Let $G$ be a PTP graph with $n$ vertices and REL $(L_1, L_2)$,
let $\P$ be a partial rectangular dual of $G$, and let $h = \abs{ \P }$.
Given a boundary path set of $(L_1, L_2)$, we can find an extension of $\P$
in $\Oh(nh)$ time.
\end{theorem}
\begin{proof}
In the proof of Theorem~\ref{clm:characterization}, we gave an algorithm that finds
for every box~$B$ obtained by the intersection of two strips 
the graph $G_B$ of vertices whose rectangles (partially) lie inside or on the boundary of~$B$.
The algorithm computes a rectangular dual~$\R_{B}$ of~$G_B$ with the algorithm of Kant and He~\cite{KH97}.
This requires $\Oh(\abs{V(G_B)})$ time per box~\cite{KH97}.
Then~$\R_{B}$ is fit into the extension built so far,
which can also be done in $\Oh(\abs{V(G_B)})$
time per box.

We now argue that $\sum_B \abs{V(G_B)} = \Oh(nh)$.
Namely, a box $B$ either lies completely inside a rectangle, in which case $\abs{V(G_B)} = 1$,
or it contains part of the boundary of every rectangle that corresponds to 
a vertex in $V(G_B)$.
For any non-fixed vertex $v$, each 
of the four boundary sides of $\R(v)$ lies either inside a single 
strip or on the boundary between two 
strips. Hence, the boundary of $\R(v)$ can lie in only $\Oh(h)$ boxes in total.
As there are $\Oh(h^2)$ boxes, we have $\sum_B \abs{V(G_B)} \in \Oh(h^2 + nh) = \Oh(nh)$.
\end{proof}

\section{Linear-Time Algorithm}
\label{sec:linear-time}
Explicitly constructing a boundary path set, as in \cref{clm:algorithm}, 
requires time proportional in the size of the set, which can however be in $\Omega(nh)$.
In this section, we show that even without an explicit construction,
we can decide if a boundary path set exists, and if so, compute an extension.
Both the decision and the computation can be done in linear time.

Our approach relies on the following observations. Suppose a boundary path set exists.
Let $v$ be a non-fixed vertex that lies on a boundary path of vertical strips $S_1, \ldots, S_k$, ordered from left to right.
Then the left boundary of $\R(v)$ lies in $S_1$ and the right boundary in $S_k$. 
Thus, to compute the x-coordinates of $\R(v)$, it suffices to know the leftmost and the rightmost boundary path on which $v$ lies.
Instead of constructing all boundary path pairs of vertical strips, 
we only construct the subgraph $H_1$ of $L_1(G)$ induced by the fixed vertices and the vertices on the boundary path pairs.
We call $H_1$ the \emph{vertical boundary graph} of $(L_1, L_2)$.
Furthermore, for each edge $e$ in $H_1$, we store the leftmost and the rightmost strip 
for which $e$ lies on a boundary path. 
Note that as $H_1$ is a subgraph of~$L_1(G)$, the size of $H_1$ is in $\Oh(n)$.
Analogously, we define $H_2$ for the horizontal~strips.

Before we show how to construct the boundary graphs $H_1$ and $H_2$, 
we prove that they suffice to compute an extension of $\P$.

\begin{lemma} \label{clm:H1}
Let $G$ be a PTP graph with $n$ vertices and REL $(L_1, L_2)$,
and let $\P$ be a partial rectangular dual of $G$.
If boundary graphs $H_1$ and $H_2$ of $(L_1, L_2)$ are given, 
then an extension of $\P$ that realizes $(L_1, L_2)$ can be computed in $\Oh(n)$ time.
\end{lemma}
\begin{proof}
We show how to compute the x-coordinates of rectangles using $H_1$; 
the y-coordinates can be computed analogously with $H_2$.
The idea is to compute a rectangular dual for each inner face of $H_1$,
which in total will yield a full rectangular dual.
Note that the boundary of each face of $H_1$ consists of two directed paths between a start and an end vertex.
Therefore, each face has a single source, a single sink, a left path, and a right path.

We distinguish two types of inner faces of $H_1$, namely,
those that contain part of the boundary of a strip and those that do not. 
A face $f$ of the former type can be identified by
the occurrence of a right-bounded vertex $v$ on the left path of $f$ 
or a left-bounded vertex $v$ on the right path of $f$
where $v$ is not the source or sink of $f$; see \cref{fig:faces}(a).
Note that in this case all inner vertices of the left path of $f$ are right-bounded
and all inner vertices of the right path of $f$ are left-bounded.
We then set the right x-coordinate of every inner vertex on the left path 
and the left x-coordinate of every inner vertex on the right path to 
the x-coordinate of the respective boundary of the strip.   

Otherwise, an inner face $f$ of $H_1$ describes a region inside a strip of $S$; see \cref{fig:faces}(b).
We define the graph $G_f$ as the subgraph of $G$ with all vertices that lie on or inside the
cycle that is defined by $f$.
By adding an outer four-cycle appropriately, we obtain a rectangular dual $\R_f$ of $G_f$ with the algorithm by Kant and He~\cite{KH97}.
We then scale $\R_f$ to the width of $S$ and set the x-coordinates
for the vertices in $G_f$ inside~$S$ accordingly,
that is, the right x-coordinate for the vertices on the left path of $f$,
the left x-coordinate for vertices on the right path of $f$,
and both x-coordinates for interior vertices of $G_f$.

After we have processed all faces, both x-coordinates of all vertices are set
since each vertex is either fixed, or has a face to the left and a face to the right,
or lies inside a face.
Since the faces are ordered from left to right in accordance with their respective strips,
the computed x-coordinates of the rectangles also form the correct horizontal adjacencies. 
We repeat this process with $H_2$ to compute the y-coordinates
and thus to obtain also the correct vertical adjacencies.
 
Processing a face $f$ takes time linear in the size of $G_f$. 
Hence, the total running time is linear in the size of $L_1(G)$ and $L_2(G)$, 
and thus in $\Oh(n)$. 
\end{proof}

\begin{figure}[tb]
  \centering
  \includegraphics{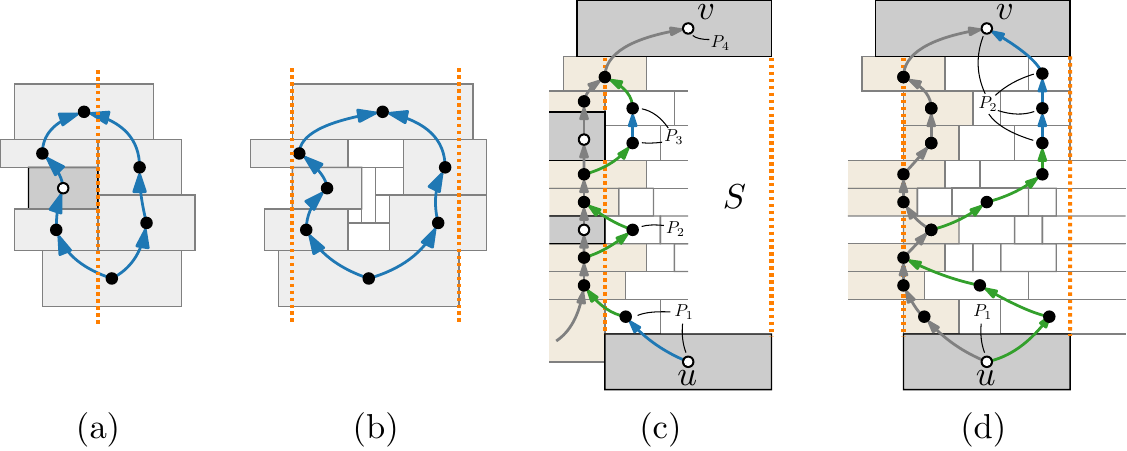}
  \caption{Correspondence between a face of $H_1$ and (a) a part of the boundary of a strip
    or (b) a region inside a strip; extending $H_1$ with (c) $\Pl(S)$ and (d) $\Pr(S)$.}
  \label{fig:faces}
\end{figure}

\begin{lemma} \label{clm:computeH1}
Let $G$ be a PTP graph with $n$ vertices and REL $(L_1, L_2)$,
and let $\P$ be a partial rectangular dual of $G$.
In $\Oh(n)$ time, we can decide whether $(L_1, L_2)$ admits a boundary path set 
with respect to $\P$ and, in the affirmative, compute boundary graphs $H_1$ and $H_2$.
\end{lemma}
\begin{proof}
As in the proof of \cref{clm:algorithm},
we focus again on vertical strips and process them from bottom-left to top-right.
Let $\Sdone$ be the strips in~$\mathcal{S}_1$ that
have already been processed.
Let $S$ be a strip with start rectangle~$\P(u)$ and end rectangle~$\P(v)$
such that every strip left of $S$ is in~$\Sdone$.
The idea is to only compute the parts of $\Pl(S)$ that
do not coincide with a boundary path of a strip $S' \in \Sdone$
and the parts of $\Pr(S)$ that do not coincide with $\Pl(S)$. 
Initially, set $H_1$ as the subgraph of $L_1(G)$ induced by the fixed vertices. 

We start with $\Pl(S)$; see \cref{fig:faces}(c). 
Observe that $\Pl(S)$ should only consist of $u$, vertices
left-bounded in~$S$, $v$,
and vertices on paths $\Pr(S')$ for $S' \in \Sdone$.
Let $P_1, \ldots, P_k$ be the subpaths induced by $u$, vertices
left-bounded in~$S$, and~$v$.
Let~$x_i$ be the first vertex of $P_i$ and $y_i$ the last.
Further let $(w_i, x_i)$ be the leftmost incoming edge of $x_i$
and let $(y_i, z_i)$ be the leftmost outgoing edge of $y_i$. 
For $i \in \set{2, \ldots, n}$, 
$w_i$ already has to be in $H_1$ and may not be a fixed vertex; 
otherwise $\Pl(S)$ does not exist.
The analogous condition needs to hold for $(y_i, z_i)$.
If this holds for each $i \in \set{1, \ldots, k}$,
we add the vertices and edges in each $P_i$
as well as the edges $(w_i, x_i)$ and $(y_i, z_i)$ to~$H_1$.
Finally, we can test that the $P_i$'s are in the correct order in $H_1$ with an st-ordering of $L_1(G)$.

Checking the existence of $\Pr(S)$ works like the construction in \cref{clm:algorithm}; see \cref{fig:faces}(d).
Recall that we extended subpaths induced by $u$, right-bounded vertices, and $v$, 
first with rightmost incoming and outgoing edges appropriately,
and then tried to join subsequent subpaths by taking the leftmost outgoing and incoming edges, respectively.
Observe that reaching $\Pl(S)$ during such an extension, now means that we encounter a non-fixed vertex in $H_1$;
see \cref{fig:faces}.  
Hence, here we stop extensions when we encounter a non-fixed vertex $v$ that is already in $H_1$.
If connecting the subpaths to each other or $H_1$ is successful, we test their order again with an st-ordering of $L_1(G)$
and finally add them to $H_1$. 

During the construction of $H_1$ we also need to label the faces with the strips they belong to.
Therefore when a subpath $P_i$, for $i \in \set{1, \ldots, k-1}$, is
added to~$H_1$ as part of a left boundary path $\Pl(S)$,
we tell its left face that is lies on the left boundary of $S$.
For any subpath added to~$H_1$ as part of a right boundary path $\Pr(S)$,
we tell its left face that it lies inside $S$.

Lastly, note that the running time is linear in the size of $H_1$, $H_2$ and $G$. 
\end{proof}

As a result of \cref{clm:H1,clm:computeH1} we get our main result, \cref{clm:main}.
\setcounter{theorem}{0}
\begin{theorem}
  \thmtext
\end{theorem}

\section{Concluding Remarks}

In this paper, we have characterized the partial rectangular duals 
that admit an extension realizing a given REL in terms of boundary path sets.
Based on this, we have given an algorithm that computes an extension,
if it exists, in time proportional to the size of the boundary path set.
We have sped up this algorithm by considering only the underlying
simple graph of a boundary path set~-- the boundary graph.

We have also formulated a system of difference constraints (a special kind of LP) that
can handle slightly more general versions of the partial rectangular dual extension problem.
Furthermore, the LP can also be used to solve
the simultaneous rectangular dual representation problem for PTP graphs with given RELs.
One can simply formulate an LP for each graph separately and then concatenate
them into a single LP where the variables for shared vertices are merged. 
As far as we know, this is the first result concerning the simultaneous representation of contact representations.
It would be interesting to see this approach applied to other contact representations.

The partial rectangular dual extension problem remains open when no
REL is specified.  Eppstein \etal~\cite{EMSV12} gave algorithms that
compute constrained area-universal rectangular duals and solved the
extension problem for RELs.  A partial rectangular dual induces a
partial REL.  Hence an extension of a partial rectangular dual~\P can
be found by computing every extension of this partial REL and by
testing for each whether it admits an extension of~\P, using our
linear-time algorithm.  There can, however, be exponentially many
extensions of a partial REL. Naturally, we are interested in a faster approach.

\pdfbookmark[1]{References}{References} 
\bibliographystyle{splncs04}
\bibliography{sources}

\end{document}